\documentclass[11pt]{scrartcl}
\usepackage{url}                      
\usepackage{tikz}
\usepackage{subfigure}
\usepackage{mathrsfs}
\usepackage{latexsym}
\usepackage{amssymb, amstext, amsgen}
\usepackage[centertags,intlimits]{amsmath}
\usepackage{amsthm}
\usepackage{a4wide}
\usepackage{lscape}
\usepackage{color}
\usepackage[utf8]{inputenc}         

\newtheorem{Theorem}{Theorem}

\newtheorem{prop}{Proposition}

\newtheorem{coro}{Corollary}

\begin{document}

\title{$S$-Packing Colorings of Cubic Graphs}

\author {Nicolas Gastineau\footnote{Author partially supported by the Burgundy Council} and Olivier Togni \\\\
\textit{LE2I UMR6306, CNRS, Arts et Métiers,}\\ \textit{Univ. Bourgogne Franche-Comté, F-21000 Dijon, France}}

\maketitle

\begin{abstract}
Given a non-decreasing sequence $S=(s_1,s_2, \ldots, s_k)$ of positive integers, an {\em $S$-packing coloring} of a graph $G$ is a mapping $c$ from $V(G)$ to $\{s_1,s_2, \ldots, s_k\}$ such that any two vertices with the $i$th color are at mutual distance greater than $s_i$, $1\le i\le k$. This paper studies $S$-packing colorings of (sub)cubic graphs. We prove that subcubic graphs are $(1,2,2,2,2,2,2)$-packing colorable and $(1,1,2,2,2)$-packing colorable. For subdivisions of subcubic graphs we derive sharper bounds, and we provide an example of a cubic graph of order $38$ which is not $(1,2,\ldots,12)$-packing colorable. 

\paragraph{Keywords:} graph, coloring, packing chromatic number, cubic graph.
\end{abstract}

\section{Introduction}

A proper coloring of a graph $G$ is a mapping which associates a color (integer) to each vertex such that adjacent vertices get distinct colors. In such a coloring, the color classes are stable sets (1-packings). As an extension, a $d$-distance coloring of $G$ is a proper coloring of the $d$-th power $G^d$ of $G$, i.e. a partition of $V(G)$ into $d$-packings (sets of vertices at pairwise distance greater than $d$).
While Brook's theorem implies that all cubic graphs except the complete graph $K_4$ of order 4 are properly 3-colorable, many authors studied 2-distance colorings of cubic graphs.

The aim of this paper is to study a mixing of these two types of colorings, i.e. colorings of (sub)cubic graphs in which some colors classes are 1-packings while other are $d$-packings, $d\ge 2$.
Such colorings can be expressed using the notion of $S$-packing coloring. For a non-decreasing sequence $S=(s_1,s_2, \ldots, s_k)$ of positive integers, an {\em $S$-packing coloring} (or simply $S$-coloring) of a graph $G$ is a coloring of its vertices with colors from $\{s_1,s_2, \ldots, s_k\}$ such that any two vertices with the $i$th color are at mutual distance greater than $s_i$, $1\le i\le k$. The color class of each color $s_i$ is thus an $s_i$-packing. The graph $G$ is {\em $S$-colorable} if there exists an $S$-coloring and it is {\em $S$-chromatic} if it is $S$-colorable but not $S'$-colorable for any $S'=(s_1,s_2, \ldots, s_j)$ with $j<k$ (notice that Goddard et al.~\cite{GX12} define differently the $S$-chromaticness for infinite graphs).

A $(d,\ldots, d)$-coloring is thus a $d$-distance $k$-coloring, where $k$ is the number of $d$ (see~\cite{KK08} for a survey of results on this invariant) while a $(1,2,\ldots, d)$-coloring is a packing coloring. The packing chromatic number $\chi_{\rho}(G)$ of
$G$ is the integer $k$ for which $G$ is $(1,\ldots,k)$-chromatic. This parameter was introduced by Goddard et al.~\cite{GoBro} under the name of {\em broadcast chromatic number} and the authors showed that deciding whether $\chi_{\rho}(G)\leq 4$ is NP-hard. A series of works~\cite{BrePa,EkPa,FiPa,FiLa,GoBro,SoPa} considered the packing chromatic number of infinite grids.
For sequences $S$ other than $(1,2,...,k)$, $S$-packing colorings were considered more recently ~\cite{GKTsub,GXsub,GX12}. Other papers are about the complexity class of the decision problem associated to the $S$-packing coloring problem~\cite{FiCo,GasCo}.


Regarding subcubic graphs, the packing chromatic number of the hexagonal lattice and of the infinite 3-regular tree is $7$ and at most $7$, respectively. Recently, Bre{\v{s}}ar et al.~\cite{BresKlav}, have proven that the packing chromatic number of some cubic graphs, namely the base-3 Sierpi{\'{n}}ski graphs, is bounded by $9$.
Goddard et al.~\cite{GoBro} asked what is the maximum of the packing chromatic number of a cubic graph of order $n$. 
For 2-distance coloring of cubic graphs, Cranston and Kim have recently shown~\cite{CK08} that any subcubic graph is $(2,2,2,2,2,2,2,2)$-colorable (they in fact proved a stronger statement for list coloring). For planar subcubic graphs $G$, there are also sharper results depending on the girth of $G$~\cite{BI11,CK08,Hav09}.

In this paper, we study $S$-packing colorings of subcubic graphs for various sequences $S$ starting with one or two `1'. We also compute the distribution of $S$-chromatic cubic graphs up to 20 vertices, for three sequences $S$. The corresponding results are reported on Tables 1, 2, and 3. They are obtained by an exhaustive search, using the lists of cubic graphs maintained by Gordon Royle~\cite{GoCu}. The paper is organized as follows: Section 2 is devoted to the study of $(1,k,\ldots,k)$-colorings of subcubic graphs for $k=2$ or $3$; Section 3 to $(1,1,2,\ldots)$-colorings; Section 4 to $(1,2,3,\ldots)$-colorings and Section 5 concludes the paper by listing some open problems.

\subsection{Notation}
To describe an $S$-coloring, if an integer $s$ is repeated in the sequence $S$, then we will denote the colors $s$ by $s_a,s_b,\ldots$.

The {\em subdivided graph} $S(G)$ of a (multi)graph $G$ is the graph obtained from $G$ by subdividing each edge once, i.e. replacing each edge by a path of length two. In $S(G)$, vertices of $G$ are called {\em original} vertices and other vertices are called {\em subdivision} vertices.
Let us call a graph {\em $d$-irregular} if it has no adjacent vertices of degree $d$. 
Notice that graphs obtained from subcubic graphs by subdividing each edge at least once are $3$-irregular graphs.

The following method (that is inspired from that of Cranston and Kim~\cite{CK08}) is used in the remainder of the paper to produce a desired coloring of a subcubic graph (except for Theorem \ref{p13}): for a graph $G$ and an edge $e=xy\in E(G)$, a {\em level ordering} of $(G,e)$ is a partition of $V(G)$ into levels $L_i=\{v\in V(G) : d(v,e)=i\}$, $0\le i \le
\epsilon(e)$, with $\epsilon(e)=\max(\{d(u,e),u\in V(G)\})\le diam(G)$. The vertices are then colored one by one, from level $\epsilon(e)$ to $1$, while preserving some properties. These properties are used at the end to allow to color the vertices $x$ and $y$ by recoloring possibly some vertices in their neighborhoods.

Two vertices $u$ and $v$ of $G$ are called \textit{siblings} if they are not adjacent, are on the same level $L_i$ for some $i\ge1$ and have a common neighbor in $L_{i-1}$.
A vertex in a level $i$ is a \textit{$k$-vertex} if it has $k$ neighbors in $L_{i}\cup L_{i+1}$. 
Notice that a 2-vertex in a subcubic graph has at most one sibling (see Figure~\ref{cousin}).
Given a (partial) coloring $c$ of $G$, let $C_1(u)=\{c(v) : uv\in E(G)\}$ and $C_2(u)=\{c(v) : d(u,v)=2, \mbox{ with }u,v \mbox{ not siblings} \}$. 

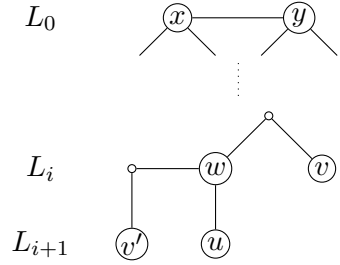
\begin{figure}
\begin{center}
 \begin{tikzpicture}
\tikzstyle{every node}=[draw,circle,fill=white,minimum size=3pt,inner sep=1pt]
\draw (-.2,0) node (x)  {$x$};
\draw (1.4,0) node (y) {$y$};
\draw (x) -- (y);
\draw[-] (x) -- +(0.5,-0.5);
\draw[-] (x)-- +(-0.5,-0.5);
\draw[-] (y) -- +(0.5,-0.5);
\draw[-] (y)-- +(-0.5,-0.5);
\draw[dotted] (0.6,-0.6) -- +(0,-0.5);
\draw (1,-1.3) node (r) {};
\draw (0.3,-2) node (w) {$w$};
\draw (1.7,-2) node (v) {$v$};
\draw (-0.8,-2) node (z) {};
\draw (-0.8,-3) node (vp) [inner sep=0pt]{$v'$};
\draw (0.3,-3) node (u) {$u$};
\draw (r) -- (v);
\draw (r) -- (w);
\draw (z) -- (w);
\draw (vp) -- (z);
\draw (w) -- (u);
\tikzstyle{every node}=[]
\node at (-2,0) (l0) {$L_0$};
\node at (-2,-2) (li) {$L_i$};
\node at (-2,-3) (lip) {$L_{i+1}$};

 \end{tikzpicture}
\end{center}
 \caption{\label{cousin}Level ordering of a subcubic graph from the edge $xy$. The vertices $v,w$ are siblings but $u,v'$ are not siblings.}
\end{figure}

\section{$(1,k,\ldots,k)$-coloring}
In this section, $(1,k,\ldots , k)$-colorings of subcubic graphs are studied for $k=2$ or $3$.

\subsection{$(1,3,\ldots,3)$-coloring}

The following proposition is used to obtain an $S$-coloring of a subdivided graph:
\begin{prop}\label{p01}

Let $G$ be a graph and $S=(s_1,\ldots, s_k)$ be a non-decreasing sequence of integers.
If $G$ is $S$-colorable then $S(G)$ is $(1, 2s_1+1,\ldots, 2s_k+1)$-colorable.

\end{prop}
\begin{proof}
Let $c$ be an $S$-coloring of $G$. Every pair of vertices $u, v\in V(G)$ such that $d(u,v)=d$ become at distance $2d$ in $S(G)$. 
Therefore, every set of vertices in $V(G)$ forming an $i$-packing also forms a $(2i+1)$-packing in $S(G)$.
Using color 1 on subdivision vertices and using the coloring $c$ (considering the sequence differently) on original vertices, we obtain a $(1, 2s_1+1,\ldots, 2s_k+1)$-coloring of $S(G)$.
\end{proof}

\begin{coro}\label{cor1}
For every subcubic graph $G$, $S(G)$ is $(1,3,3,3)$-colorable.
\end{coro}
\begin{proof}
Brooks' theorem asserts that every subcubic graph except $K_4$ is $(1,1,1)$-colorable. Hence, by Proposition \ref{p01}, every subcubic graph $G$ except $K_4$ is such that $S(G)$ is $(1,3,3,3)$-colorable.
We define a $(1,3,3,3)$-coloring of $S(K_4)$ as follows: let $\gamma:E(K_4)\rightarrow \{a,b,c\}$ be a proper edge 3-coloring of $K_4$. Put color 1 on all four original vertices of $K_4$ and put color $3_{\gamma(e)}$ on each subdivision vertex corresponding to edge $e$ of $K_4$.
\end{proof}

Goddard et al.~\cite{GoBro} characterized $(1,3,3)$-colorable graphs as the graphs obtained from any bipartite multigraph by subdividing it and adding leaves on original vertices. Therefore, there are many subdivided subcubic graphs that are not $(1,3,3)$-colorable (for instance $S(C_3)=C_6$), showing that the bound of Corollary \ref{cor1} is tight in a certain sense.

\subsection{$(1,2,\ldots,2)$-coloring}
Notice that a vertex of degree at least $3$ in a $(1,2,2)$-colored graph can not be colored by $1$. Thus, a  $(1,2,2)$-colorable graph does not contain three vertices of degree larger than 2 at mutual distance at most $2$ and in particular no cubic graph is $(1,2,2)$-colorable.
However, there exist $(1,2,2)$-colorable subcubic graphs and it has been recently proved~\cite{GasCo} that determining if a subcubic bipartite graph is $(1,2,2)$-colorable is NP-complete.

\begin{Theorem}\label{p12}
 Every subcubic graph is $(1,2,2,2,2,2,2)$-colorable.
\end{Theorem}
\begin{proof}
Let $G$ be a subcubic graph and let $e=xy$ be any edge of $G$. Define a level ordering $L_i$, $0\le i \le
r=\epsilon(e)$, of $(G,e)$. 

 
We first construct a coloring $c$ of the vertices of $G$ from level $r$ to $1$ and with colors from the set $C=\{1,2_a,2_b,2_c,2_d, 2_e, 2_f\}$, that satisfies the following properties:
\begin{enumerate}
\item[i)] color 1 is used as often as possible, i.e. when coloring a vertex $u$, if no neighbor is colored 1, then $u$ is colored 1;
\item[ii)] if $u$ is colored by $2$, then there is a subsidiary color $\tilde{c}(u)\in C$ different from $c(u)$ such that $\tilde{c}(u)\not\in C_1(u)\cup C_2(u)$, but with
possibly $\tilde{c}(u)=c(v)$ if $u$ and $v$ are siblings.
\end{enumerate}

The set $L_r$ induces a disjoint union of paths and cycles in $G$. Since paths and cycles are $(1,2,2,2)$-colorable, we are able to construct a coloring of the vertices of $L_{r}$ as follows. Start by coloring each path/cycle with colors $\{1,2_a,2_b,2_c\}$ with color $1$ used as often as possible. We suppose that color $1$ is used for the end-vertices (vertices of degree 1) of the paths except in the case the path has length $1$.
Remark that a vertex of a cycle of $L_r$ or a vertex of degree $2$ inside a path of $L_r$ can be at distance $2$ of at most one vertex in $L_r$ outside this cycle/path. A vertex of a path of length $1$ of $L_r$ can be at distance $2$ of at most two vertices in $L_r$ outside this path.
We then recolor each vertex $u$ of a path of length $1$ in $L_r$ with color from $\{2_a,2_b,2_c\}$ by a color from $\{2_d,2_e,2_f\}$ not given to vertices at distance at most $2$ from $u$ and we set $\tilde{c}(u)=2_a$.
Then, for each pair of vertices $u,v$ in different paths/cycles at distance 2 both colored by $2_a$ ($2_b$ or $2_c$ respectively), set $c(u)=2_d$ ($2_e$ or $2_f$, respectively). Afterwards, for every vertex $u$ of color $2_a$ ($2_b$, $2_c$, $2_d$, $2_e$ or $2_f$, respectively), set $\tilde{c}(u)=2_d$ ($2_e$, $2_f$, $2_a$, $2_b$ or $2_c$, respectively). Then, the produced coloring is a partial $(1,2,2,2,2,2,2)$-coloring of $G$ and Property ii) is satisfied.

Assume that we have already colored all vertices of $G$ of levels from $r$ to $i+1$ and that we are going to color
vertex $u\in L_i$, $1\le i\le r-1$. 
If $1\not\in C_1(u)$ then set $c(u)=1$ (Property i) is then satisfied).
If $u$ is a $0$-vertex, then it has been colored by $1$. If $u$ is a $1$-vertex, then $|C_1(u)\cup C_2(u)|\le 3$ and $u$ has at most two siblings $v$ and $v'$. In this case, we can set a color and a subsidiary color to $u$ from the set $C\setminus ( C_1(u)\cup C_2(u)\cup \{c(v), c(v') \} )$. Therefore, we now suppose that $u$ is a 2-vertex.

If $1\in C_1(u)$, then let $u_1$ be the neighbor of $u$ of color $1$ and let $u_2$ be the other neighbor of $u$,
if any. By construction, either $c(u_2)=1$ or $1\in C_1(u_2)$, hence $|C_1(u)\cup C_2(u)|\le 5$. 
In that case there are at least two colors $\{2_\alpha, 2_\beta\}\subset C\setminus \{C_1(u)\cup C_2(u)\}$ for some $\alpha,\beta\in \{a,\ldots,f\}$, with
possibly, if $u$ has a sibling $v$, $2_\beta=c(v)$. Then set $c(u)=2_\alpha$ and $\tilde{c}(u)=2_\beta$ (Property ii) is then satisfied). Figure~\ref{f1p12} illustrates this case.

\begin{figure}
\begin{center}
 \begin{tikzpicture}
 \tikzstyle{every node}=[draw,circle,fill=white,minimum size=3pt,inner sep=1pt]
\draw (0,0) node (r)  {}
      (-1,-1) node (u) [label=left:$2_f$:$2_e$] {$u$}
      (1,-1) node (v)[label=left:$2_e$] {$v$};
 \draw  [dashed]   (r) --++ (u) 
                   (r) --++ (v);
\draw (u) -- ++(220:1.5cm) node (u1) [inner sep=0pt,label=left:$1$] {$u_1$}
     (u1)-- ++(240:1.2cm) node () [label=left:$2_a$]{}
     (u1)-- ++(300:1.2cm) node () [label=left:$2_b$]{}
     (u)-- ++(320:1.5cm) node (u2) [inner sep=0pt,label=left:$2_d$]{$u_2$}
     (u2)-- ++(240:1.2cm) node ()[label=left:$1$] {}
     (u2)-- ++(300:1.2cm) node () [label=left:$2_c$]{};
 
 \end{tikzpicture}
\end{center}
\caption{\label{f1p12}A configuration in the proof of Theorem~\ref{p12}, when coloring vertex $u$. The label $2_f$:$2_e$ on $u$ means that $c(u)=2_f$ and $\tilde{c}(u)=2_e$.}
\end{figure}
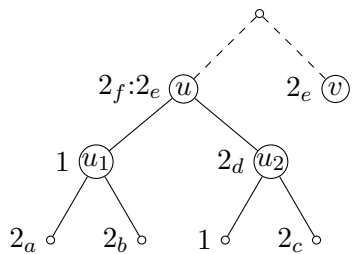

Finally, it remains to color vertices of $L_0$, i.e., $x$ and $y$. If $1\in C_1(x)\cap C_1(y)$ then, by Property i), the
neighbor $x_2$ of $x$ colored by $2$, if any, has a neighbor of color 1 and the same goes for $y$, with $y_2$ being the neighbor of $y$ colored by $2$, if any. Hence $|C_1(x)\cup C_2(x)|\le
5$, $|C_1(y)\cup C_2(y)|\le5$ and there remains at least two colors $2$ available for $x$ and two colors $2$ available for $y$. Therefore $x$ and $y$ can be assigned a different color $2$.

If $1\in C_1(x)$ but $1\not \in C_1(y)$ (or $1\in C_1(y)$ but $1\not \in C_1(x)$, by symmetry), then set $c(y)=1$. We recall that $x_2$ is the neighbor of $x$ not colored by $1$. If $C_1(x)\cup C_2(x)=C$, then set $c(x)=c(x_2)$ and $c(x_2)=\tilde{c}(x_2)$, else give to $x$ an available color.

Otherwise, $1\not\in C_1(x)\cup C_1(y)$. Then set $c(y)=1$ and we show that there is always a color 2 to assign to $x$.
If $|C_1(x)\cup C_2(x)|\le 6$, then there is a color available for $x$. Else, let $x_1$, $x_2$ be the two neighbors of
$x$ other than $y$ and let $x'_1$ ($x'_2$, respectively) be the neighbor of $x_1$ ($x_2$, respectively) colored 2 other
than $x$ (no more than one, as $x_1$ and $x_2$ both have a neighbor colored by $1$). 
Suppose, without loss of generality, that $c(x_1)=2_a$, $c(x_2)=2_b$, $c(x'_1)=2_c$ and $c(x'_2)=2_d$.
If $\tilde{c}(x_1)\in \{2_d,2_e,2_f\}$ then recolor $x_1$ by its subsidiary color $\tilde{c}(x_1)$ and set $c(x)=2_a$. Similarly, if
$\tilde{c}(x_2)\in \{2_c,2_e,2_f\}$ then recolor $x_2$ by its subsidiary color $\tilde{c}(x_2)$ and set $c(x)=2_b$. Else,
$\tilde{c}(x_1)=2_b$ and $\tilde{c}(x_2)=2_a$. Recolor $x'_1$ by its subsidiary color $\tilde{c}(x'_1)$ and set $c(x)=2_c$. If
$\tilde{c}(x'_1)=2_a=c(x_1)$, then switch the colors of $x_1$ and $x_2$ (this is possible since $\tilde{c}(x_1)=c(x_2)$
and $\tilde{c}(x_2)=c(x_1)$). Figure~\ref{f3p12} illustrates this case.

\begin{figure}
\begin{center}
 \begin{tikzpicture}
 \tikzstyle{every node}=[draw,circle,fill=white,minimum size=3pt,inner sep=1pt]
\draw (0,0) node (x)  {$x$}
        -- ++(0:3.2cm) node (y)  {$y$}
        -- ++(300:1.5cm) node [label=left: $2_f$](y2) {}
     (y)-- ++(240:1.5cm) node [label=left:$2_e$](y1) {}
     (x)-- ++(300:1.5cm) node [inner sep=0pt,label=left:$2_b$:$2_a$](x2) {$x_2$}
     (x)-- ++(240:1.5cm) node [inner sep=0pt,label=left:$2_a$:$2_b$](x1) {$x_1$}
     (x1)-- ++(250:1.2cm) node [inner sep=0pt,label=left:$2_c$:$2_a$](x11) {$x'_1$}
     (x1)-- ++(290:1.2cm) node [label=left:$1$](x12) {}
     (x2)-- ++(290:1.2cm) node [inner sep=0pt,label=right:$2_d$](x22) {$x'_2$}
     (x2)-- ++(250:1.2cm) node [label=left:$1$](x21) {};
 \draw (6.5,0) node (x) [label=above:$2_c$] {$x$}
         -- ++(0:3.2cm) node[label=above:$1$] (y)  {$y$}
        -- ++(300:1.5cm) node [label=left: $2_f$](y2) {}
     (y)-- ++(240:1.5cm) node [label=left:$2_e$](y1) {}
     (x)-- ++(300:1.5cm) node [inner sep=0pt,label=left:$2_a$](x2) {$x_2$}
     (x)-- ++(240:1.5cm) node [inner sep=0pt,label=left:$2_b$](x1) {$x_1$}
     (x1)-- ++(250:1.2cm) node [inner sep=0pt,label=left:$2_a$](x11) {$x'_1$}
     (x1)-- ++(290:1.2cm) node [label=left:$1$](x12) {}
     (x2)-- ++(290:1.2cm) node [inner sep=0pt,label=right:$2_d$](x22) {$x'_2$}
     (x2)-- ++(250:1.2cm) node [label=left:$1$](x21) {};      
 \end{tikzpicture}
\end{center}
 \caption{\label{f3p12}A configuration in the proof of Theorem~\ref{p12}, before (on the left) and after (on the right) coloring $x$ and $y$.}
\end{figure}
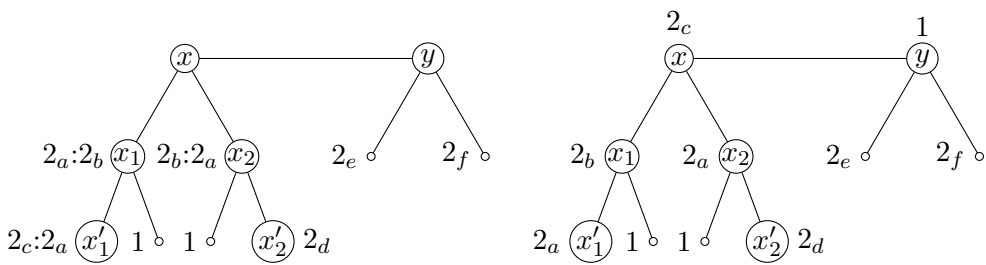

Therefore, we obtain, in all cases, a $(1,2,2,2,2,2,2)$-coloring of $G$.
\end{proof}

The Petersen graph is an example of cubic graph which is not $(1,2,2,2,2,2)$-colorable, showing that the result of
Theorem~\ref{p12} is tight in a certain sense. However, the experiments reported on Table~\ref{t2} suggest that the Petersen graph could be the only non
$(1,2,2,2,2,2)$-colorable subcubic graph. These experiments have been done using a computer to check every possible color configuration and by verifying that it is a $S$-coloring, for sequences $S$ with $s_1=1$, $|S|\ge4$ and $s_i=2$, for $i\ge2$ . We have considered the list of cubic graphs from Gordon Royle \cite{GoCu} .

\begin{table}[!ht]
\centering
\begin{tabular}{c|cccc}
$n\backslash S$&	$(1,2,2,2)$&	$(1,2,2,2,2)$&	$(1,2,2,2,2,2)$&	$(1,2,2,2,2,2,2)$\\\hline
4&  1 & 0 & 0 & 0\\
6&  1 & 1 & 0 & 0\\
8&		2&     1&	2&	0\\
10&		11&	   7&	0&	1\\
12&		11&	   74&	0&	0\\
14&		254&	  250&	5&	0\\
16&		1031&  3017& 12&	0\\
18&		15960&	25297&	44&	0\\
20&		178193&	332045&	251&	0\\
22&		2481669&	4835964&	1814&	0\\
\end{tabular}
\caption{\label{t2}Number of $S$-chromatic cubic graphs of order $n$ up to $22$.}
\end{table}

%

Furthermore, as the following proposition shows, even some bipartite cubic graphs are not $(1,2,2,2,2,3)$-colorable.
\begin{prop}
There exist bipartite cubic graphs that are not $(1,2,2,2,2,3)$-colorable.
\end{prop}
\begin{proof}
The cubic graph depicted in Figure~\ref{fbip14} is bipartite and is $(1,2,2,2,2,2)$-colorable, as shown on the figure.
Let $(A,B)$ be the two subsets of vertices that form a bipartition of this graph. Suppose this graph is $(1,2,2,2,2,3)$-colorable and let $c$ be a $(1,2,2,2,2,3)$-coloring and $X_1$ be the set of vertices colored 1.
Remark that the cardinality of any $2$-packing is at most $2$ and that any pair of vertices $(u,v)$ included in $A$ or in $B$ is such that $d(u,v)\le 2$.
We have $|X_1|\ge 5$, as at most one vertex can be colored by $3$ (since the diameter of the graph is $3$) and at most two vertices can be colored the same color $2$.

First, if $X_1\subseteq A$ or $X_1\subseteq B$, then each remaining vertex should be colored differently in the other partition, which is impossible since $|A|=|B|=7$.

Second, if there are vertices colored by $1$ in $A$ and $B$, then the only possibility in order to have $|X_1|\ge 5$ is to have one vertex colored by $1$ in one partition and four vertices colored 1 in the other partition (since othewise we obtain adjacent vertices of color $1$). Suppose, without loss of generality, that $|X_1\cap A|=1$ and $|X_1\cap B|=4$. Exactly three vertices are not colored 1 in $B$. Consequently, only three pairs of vertices can have the same color $2$ and the three vertices not colored by $1$ in $A$ cannot be all colored with the remaining colors 2 and 3.
\end{proof}

\begin{figure}[ht]
\begin{center}
\begin{tikzpicture}[scale=1]
 \tikzstyle{every node}=[draw,circle,fill=white,minimum size=3pt,inner sep=1pt]
\draw (0,0) node (r)  [label=above:$1$]{}
        -- ++(-2.8cm,-1cm) node (f1) [label=left:$2_a$] {}
       (r) -- ++(0,-1cm) node (f2) [label=left:$2_b$] {}
       (r) -- ++(2.8cm,-1cm) node (f3) [label=right:$2_c$] {}
       (f1) -- ++(240:1.3cm) node (f11) [label=left:$1$] {}
      (f1) -- ++(300:1.3cm) node (f12) [label=left:$2_b$] {}
      (f2) -- ++(240:1.3cm) node (f21) [label=left:$2_a$] {}
      (f2) -- ++(300:1.3cm) node (f22) [label=left:$2_c$] {}
      (f3) -- ++(240:1.3cm) node (f31) [label=right:$2_d$] {}
      (f3) -- ++(300:1.3cm) node (f32) [label=right:$2_e$] {}
      (f11) -- ++(1cm,-1.8cm) node (f111) [label=left:$2_e$] {}
      (f11)-- ++(2.5cm,-1.8cm) node (f112) [label=left:$2_d$] {}
      (f22)-- ++(0.5cm,-1.8cm) node (f221) [label=right:$1$] {}
      (f32)-- ++(-1cm,-1.8cm) node (f321) [label=right:$1$] {}
      (f111)-- ++ (f21)
      (f111)--++(f31)
      (f112)--++ (f22)
      (f112)--++ (f32)
      (f221)--++ (f12)
      (f221)--++ (f31)
      (f321)--++ (f12)
      (f321)--++ (f21);
\end{tikzpicture}
\end{center}
\caption{\label{fbip14} A cubic bipartite $(1,2,2,2,2,2)$-chromatic graph of order 14.}
\end{figure}
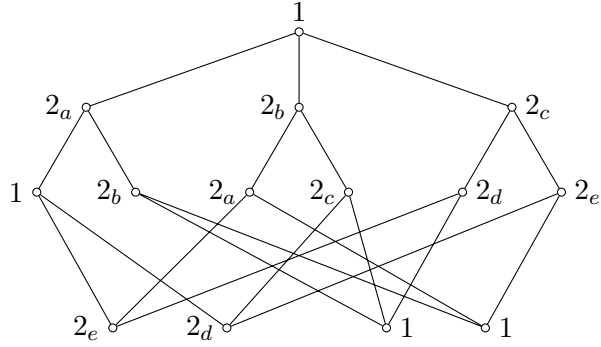

The next results show that there are sub-families of subcubic graphs that can be colored with fewer colors.

\begin{Theorem}\label{p1222}
Every 3-irregular subcubic graph is $(1,2,2,2)$-colorable. 
\end{Theorem}

\begin{proof}
Let $G$ be a 3-irregular graph and let $e=xy$ be any edge of $G$ such that $x$ and $y$ are both of degree at most $2$.
If no such edge exists, then the graph is the subdivision $S(H)$ of some subcubic graph $H$ where leaves could be added on original vertices of degree 2 and thus $G$ is $(1,3,3,3)$-colorable by Corollary \ref{cor1}.
Define a level ordering $L_i$, $0\le i \le r=\epsilon(e)$, of $(G,e)$.

We construct a coloring $c$ of the vertices of $G$ from level $r$ to $1$ and with colors from the set $\{1,2_a, 2_b,2_c\}$, that satisfies the following properties:
\begin{enumerate}
\item[i)] color $1$ is used as often as possible for vertices of degree at most 2, i.e. when coloring a vertex $u$ of degree at most $2$, if no neighbor is colored by $1$, then $u$ is colored by $1$;
\item[ii)] every vertex of degree $2$ is colored by $1$ when first coloring vertices of $L_i$, except if the connected component containing this vertex in $L_i$ is a path of order 2 (in which case one of the two vertices is colored 1).
\end{enumerate}

The set $L_{r}$ induces a disjoint union of paths of order at most 3 in $G$. Since paths are $(1,2,2)$-colorable, the vertices of $L_{r}$ can be $(1,2,2,2)$-colored.
Moreover, in every path of order 3 in $L_r$, the central vertex has degree 3, thus a color $1$ could be given to every vertex of degree $2$. If the path is of order 2, one of its end-vertices is colored by $1$.
Thus, Properties i) and ii) are satisfied.

Assume that we have already colored all vertices of $G$ of levels from $r$ to $i+1$ and that we are going to color vertex $u\in L_i$, $1\le i\le r-1$. We consider two cases depending on the degree of $u$:

\noindent\textbf{Case 1.}  $u$ is of degree 3.

If $1\notin C_1(u)$, then $u$ can be colored by $1$.
Let $u_1$ and $u_2$ be the colored neighbors of $u$, with $c(u_1)=1$.
By Property i), either $u_2$ or a colored neighbor of $u_2$ has color 1. Hence, we have $|C_1(u)\cup C_2(u)|\le 3$ and $u$ can be colored some color $2$.

\noindent\textbf{Case 2.}  $u$ is of degree 1 or 2.

If $u$ has degree 1, then we can color $u$ by 1. If $1\notin C_1(u)$, then we can set $c(u)=1$.
Otherwise, let $u_1$ be the colored neighbor of $u$, if any. If $u_1$ is of degree 3, let $u_{1,1 }$ and $u_{1,2}$ be the colored neighbor of $u$, let $u_{1,1,1}$ be the neighbor of $u_{1,1}$ different from $u$ and let $u_{1,2,1}$ be the neighbor of $u_{1,2}$ different from $u$.
Since $1\in C_1(u)$, $c(u_1)=1$ and thus $c(u_{1,1})\ne 1$ and $c(u_{1,2})\ne 1$. Therefore, by Property i), $c(u_{1,1,1})=c(u_{1,2,1})=1$.
Thus, $u_1$ can be recolored by some color $2$ and we can set $c(u)=1$.
If $u_1$ is of degree at most 2, then, since $|C_1(u_1)\cup C_2(u_1)|\le 3$, we can recolor $u_1$ by a color $2$.
Thus, we can set $c(u)=1$.
\medskip

Finally, it remains to color vertices of $L_0$, i.e. $x$ and $y$.
Let $x_1$ be the possible neighbor of $x$ different from $y$ and let $y_1$ be the possible neighbor of $y$ different from $x$. We consider three cases that cover all the possibilities by symmetry:

\noindent\textbf{Case 1.} $x_1$ and $y_1$ both have degree 3.

If their neighbors different from $x$ and $y$ are not adjacent between them, then, by Property ii),
these vertices have color $1$ (they have not been recolored when coloring the vertices of $L_1$ since $x_1$ and $y_1$ have degree 3) and $x_1$ and $y_1$ have some color $2$.
Thus we can set $c(x)=1$ and some color $2$ to $y$, as $|C_1(y)\cup C_2(y)|\le 3$.
Suppose that the two vertices of $N(x_1)\setminus\{x\}$ are each adjacent to a different vertex of $N(y_1)\setminus \{y\}$.
Thus, $G$ is the graph $G_{(1,2,2,2)}$ from Figure \ref{f1p15} and is $(1,2,2,2)$-colorable.
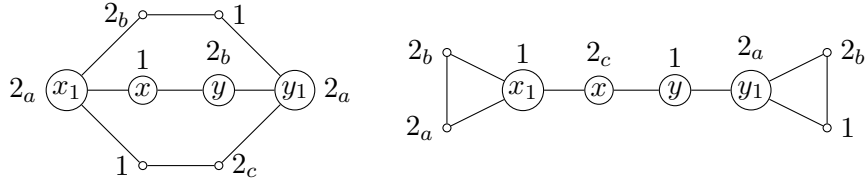
\begin{figure}
\begin{center}
 \begin{tikzpicture}
\tikzstyle{every node}=[draw,circle,fill=white,minimum size=3pt,inner sep=1pt]
\draw (0,0) node (x1)[label=left:$2_a$]  {$x_1$};
\draw (1,0) node (x) [label=north:$1$] {$x$};
\draw (2,0) node (y) [label=north:$2_b$] {$y$};
\draw (3,0) node (y1)[label=right:$2_a$]  {$y_1$};
\draw (1,1) node (x11)[label=left:$2_b$]   {};
\draw (2,1) node (y11)[label=right:$1$]  {};
\draw (1,-1) node (x12)[label=left:$1$]   {};
\draw (2,-1) node (y12)[label=right:$2_c$]  {};
\draw (x1) -- (x);
\draw (x) -- (y);
\draw (y1) -- (y);
\draw (x1) -- (x11);
\draw (y1) -- (y11);
\draw (x1) -- (x12);
\draw (y1) -- (y12);
\draw (x11) -- (y11);
\draw (x12) -- (y12);

\draw (5,-0.5) node (x11')[label=left:$2_a$]  {};
\draw (5,0.5) node (x12')[label=left:$2_b$]  {};
\draw (6,0) node (x1')[label=north:$1$]  {$x_1$};
\draw (7,0) node (x')[label=north:$2_c$]  {$x$};
\draw (8,0) node (y')[label=north:$1$]  {$y$};
\draw (9,0) node (y1')[label=north:$2_a$]  {$y_1$};
\draw (10,-0.5) node (y11')[label=right:$1$]  {};
\draw (10,0.5) node (y12')[label=right:$2_b$]  {};

\draw (x1') -- (x');
\draw (x') -- (y');
\draw (y1') -- (y');
\draw (x1') -- (x11');
\draw (y1') -- (y11');
\draw (x1') -- (x12');
\draw (y1') -- (y12');
\draw (x11') -- (x12');
\draw (y11') -- (y12');

\end{tikzpicture}
\end{center}
\caption{\label{f1p15} The graphs $G_{(1,2,2,2)}$ (on the left) and $G'_{(1,2,2,2)}$ (on the right) from Proposition~\ref{p1222}.}
\end{figure}
Suppose that the two vertices of $N(x_1)\setminus \{x\}$ are adjacent and that the two vertices of $N(y_1)\setminus \{y\}$ are adjacent.
In this case, $G$ is the graph $G'_{(1,2,2,2)}$ from Figure \ref{f1p15} and is $(1,2,2,2)$-colorable.
Suppose that only one vertex  of $N(x_1)\setminus \{x\}$ is adjacent to a vertex of $N(y_1)\setminus \{y\}$.
Let $x_{1,1}$ and $y_{1,1}$ be these two adjacent neighbors, the other neighbors are colored by $1$ by Property ii).
One of these two vertices is colored by $1$ and the other one is colored by $2$. Suppose without loss of generality that $c(x_{1,1})=1$.
Hence, we have $|C_1(x)\cup C_2(x)|\le 3$ and we can color $x$ by a color $2$ and set $c(y)=1$.
Suppose now that the two vertices of $N(x_1)\setminus \{x\}$ are adjacent and that the two vertices of $N(y_1)\setminus \{y\}$ are not adjacent.
By Property ii), the neighbors of $y_1$ have color $1$ and $y_1$ has some color $2$. In this case, we color $y$ by $1$, $x$ by another color $2$, $x_1$ by $1$ and the two adjacent neighbors of $x_1$ by colors from $C\setminus \{1, c(x)\}$.

\noindent\textbf{Case 2.} $x_1$ has degree at most $2$ and $y_1$ has degree $3$.

Since $|C_1(y_1)\cup C_2(y_1)|\le 3$, then $y_1$ can be recolored by some color $2$. 
If $x_1$ and $y_1$ are not adjacent, then $x_1$ is colored by $1$ by Property ii). Otherwise ($x_1$ and $y_1$ are adjacent), we can color $x_1$ by $1$ since $y_1$ has not been colored by $1$.
Afterward, we set $c(y)=1$, and since $|C_1(x)\cup C_2(x)|\le 3$ we can set a color $2$ to $x$.

\noindent\textbf{Case 3.} $x_1$ and $y_1$ are both of degree at most 2. 

If $x_1$ and $y_1$ are adjacent, then the graph is $C_4$ which is trivially $(1,2,2,2)$-colorable.
If $x_1$ and $y_1$ are not adjacent, then, by Property ii), they both have color $1$, $|C_1(x)\cup C_2(x)|\le 2$ and $|C_1(y)\cup C_2(y)|\le 2$.
Thus, we can set some colors $2$ to $x$ and $y$.

Therefore, we obtain in all cases a $(1,2,2,2)$-coloring of $G$.
\end{proof}

Remark that the 5-cycle $C_5$ is 3-irregular and is not $(1,2,2)$-colorable, hence the result of Theorem~\ref{p1222} is tight in a certain sense.
However, there are 3-irregular subcubic graphs that are $(1,2,2,3)$-colorable. The graphs from Figure~\ref{f1p15} are such examples (the color $2_c$ can be replaced by color 3).

We end this section with some results on subdivided graphs. Let $\delta(G)$ be the minimum degree of $G$. In the following theorem, we suppose that $\delta(G)\ge 3$ in order to avoid cycles with few vertices of degree at least 3 (every cycle, except $C_5$, is $(1,2,2)$-colorable).
\begin{prop}\label{p02}
For every graph $G$ with $\delta(G)\ge 3$, if $S(G)$ is $(1,2,2)$-colorable, then $G$ is bipartite.
\end{prop}
\begin{proof}
Suppose $S(G)$ is $(1,2,2)$-colorable and $G$ contains an odd cycle. In every $(1,2,2)$-coloring of a graph, every vertex of degree at least $3$ should be colored some color $2$
(if a vertex of degree at least $3$ is colored by $1$, the coloring cannot be extended to the neighbors of this vertex).
Therefore, if $G$ contains an odd cycle, then $S(G)$ contains a cycle with an odd number of vertices of degree $3$ and the colors $2_a$ and $2_b$ are not sufficient to alternately color these vertices.
Hence $S(G)$ is not $(1,2,2)$-colorable.
\end{proof}
Since every bipartite graph $G$ is $(1,1)$-colorable, by Proposition \ref{p01}, $S(G)$ is $(1,3,3)$-colorable (and also $(1,2,2)$-colorable and $(1,2,3)$-colorable). Thus, we obtain the following corollary.
\begin{coro}
For every graph $G$ with $\delta(G)\ge 3$,\\
\\
$S(G)$ is $(1,2,2)\text{-colorable} \Leftrightarrow S(G)$ is $(1,2,3)\text{-colorable} \Leftrightarrow S(G)$ is $(1,3,3)\text{-colorable}\Leftrightarrow G$ is $\text{bipartite}$.
\end{coro}

\section{$(1,1,2,\ldots)$-coloring}
Similarly with $(1,2,2)$-coloring, it has been recently proved~\cite{GasCo} that determining if a subcubic or a cubic graph is $(1,1,2)$-colorable is NP-complete.
Remind that bipartite graphs are $(1,1)$-colorable. For non-bipartite subcubic graphs, we prove the following result using a different argument than for the previous theorems. Using the method of proofs of Theorems \ref{p12} and \ref{p1222} allows to show that subcubic graph are $(1,1,2,2,2,2)$-colorable.
\begin{Theorem}\label{p13}
Every subcubic graph is $(1,1,2,2,2)$-colorable.
\end{Theorem}

\begin{proof}
Let $G$ be a subcubic graph.
For a vertex $u$ of $V(G)$, we denote by $B_2(u)$ the set $\{v\in V(G)|\  d_G (u,v)\le 2,\ u\neq v \}$.
A set $X\subseteq V(G)$ is an \emph{odd-cut set} if the graph induced by $V(G)\setminus X$ is a bipartite graph. Such set can be obtained by removing one vertex per odd cycle from $V(G)$.
Let $G_X$ be the graph with vertex set $X$ and edge set $\{ uv |\ u,\ v\in X,\ d_{G}(u,v)\le 2,\ u\neq v \}$.
The proof consists of proving the existence of an odd-cut set $X$ such that $G_X$ is a subcubic graph (not necessarily connected) and $G_X$ has no connected component isomorphic to $K_4$.

If $G_X$ is subcubic and has no connected component isomorphic to $K_4$, then we can construct a coloring $c$ of $G$ with colors from the set $C=\{1_a,1_b,2_a,2_b,2_c\}$ as follows.
Since the graph induced by $V(G)\setminus X$ is bipartite, we can color the vertices of $V(G)\setminus X$ with colors $1_a$ and $1_b$.
By Brook's Theorem, if $G_X$ is subcubic and has no connected component isomorphic to $K_4$, then there exists a proper vertex-coloring $c'$ of $G_X$ with colors from the set $\{a,b,c\}$.
For a vertex $u\in X$, we define $c$ by $c(u)=2_{c'(u)}$.

For a cycle $C$ of $G$, we denote by $V(C)$ the set of vertices in $C$. 
For an odd-cut set $X\subseteq V(G)$ and a vertex $u\in X$, we denote by $C_u$ a cycle such that $X\cap V(C_u)=\{u\}$.
In the rest of the proof, when $G_X$ is supposed to have an induced subgraph isomorphic to $K_4$, we denote by $u_1$, $u_2$, $u_3$ and $u_4$ the vertices of this induced subgraph and by $U$ the set $\{u_1,u_2,u_3,u_4\}$. For two vertices $u_i$ and $u_j$, $1\le i<j\le 4$, we denote by $u_{ij}$ a common neighbor of $u_i$ and $u_j$ in $V(G)$ when $d_G (u_i,u_j)=2$. When $u_{ij}$ exists, we denote by $u'_{ij}$, the possible vertex in $N(u_{ij})\setminus\{u_i,u_j\}$. \\

\textbf{Claim 1.}
For an odd-cut set $X$ in $G$ with minimum cardinality, the following properties hold:
\begin{enumerate}
\item[(A1)] for every vertex $u\in X$, there exists at least one cycle $C$ such that $V(C)\cap X=\{u\}$ and consequently $|N(u)\cap X|\le 1$;
\item[(A2)] for every two vertices $u,v\in X$ such that $uv\in E(G)$, we have $d_{G_X}(u)\le 3$ and $d_{G_X}(v)\le 3$;
\item[(A3)] for every vertex $u\in X$, $d_{G_X}(u)\le 4$;
\item[(A4)] if there does not exist three vertices $u,v,w \in X$ such that $N(u)\cap N(v) \cap N(w)\neq \emptyset$ and $\max(d_{G_X}(u),d_{G_X}(v),d_{G_X}(w) )\ge 3$, then for every vertex $x\in X$, $d_{G_X}(x)\le 3$.
\end{enumerate}
\textit{Proof.} The left part of Figure \ref{f1p11222a} illustrates a vertex $u\in X$ satisfying $d_{G_X}(u)= 4$. 

(A1): If such cycle did not exist, then $X\setminus\{u\}$ would be also an odd-cut set, contradicting its minimality. Moreover, since $u$ has two neighbors in $V(C_u)$, $|N(u)\cap X|\le 1$.

(A2): Let $u$ and $v$ be two adjacent vertices of $X$.
By Property (A1), there exist two pairs of distinct vertices $\{u_1,u_2\}$ and $\{v_1,v_2\}$ such that $\{u_1,u_2\}\subseteq N(u)\cap V(C_u)$ and $\{v_1,v_2\}\subseteq N(v)\cap V(C_v)$. Notice that both $u_1$ and $u_2$ have at least two neighbors in $V(C_u)$. Then, $|B_2(u)\setminus (V(C_u)\cup V(C_v))|\le 3 $.
Since the same property holds for $C_v$, we have $d_{G_X}(u)\le 3$ and $d_{G_X}(v)\le 3$.

(A3) and (A4): Let $u$ be a vertex of $X$. By Property (A1), there exist two distinct vertices $u_1$ and $u_2$ such that $\{u_1,u_2\}\subseteq N(u)\cap V(C_u)$. As previously, $u_1$ and $u_2$ have at least two neighbors in $V(C_u)$. Let $u_3$ be the possible vertex in $N(u)\setminus V(C_u)$. By Property (A2), if $u_3\in X$, then $d_{G_X}(u)\le 3$. Thus, we can now suppose that $u_3\notin X$.
Since $|(N(u_3)\setminus \{u\})\cap X|\le 2$, we have $|B_2(u)\cap X|\le 4$ and consequently $d_{G_X}(u)\le 4$.
Moreover, if there exists a vertex $u$ such that $d_{G_X}(u)=4$, then we have $|(N(u_3)\setminus \{u\})\cap X|= 2$ and $u_3$ is a common neighbor of three distinct vertices of $X$. Consequently, the contrapositive of Property (A4) follows.
\begin{flushright} $\square$ \end{flushright}

\textbf{Claim 2.}
For an odd cut set $X$ of minimum cardinality minimizing $|E(G_X)|$, the following properties hold:\\
\begin{enumerate}
\item[(B1)] there does not exist three vertices $u,v,w\in X$ such that $N(u)\cap N(v) \cap N(w)\neq \emptyset$ and $\max(d_{G_X}(u),d_{G_X}(v),d_{G_X}(w) )\ge 3$;
\item[(B2)] if there is an induced $K_4$ in $G_X$, then at most two vertices of $U$ are adjacent in $G$;
\item[(B3)] if there is an induced $K_4$ in $G_X$ and two vertices of $U$ are adjacent in $G$, then no vertex of $U$ is in an induced triangle in $G$;
\item[(B4)] if there is an induced $K_4$ in $G_X$ and two vertices of $U$ are adjacent in $G$, then for every vertex $v\in \{u_{12},u_{13},u_{14},u_{23},u_{24},u_{34}\}$, we have $|B_2(v)\cap X|\le 4$;
\item[(B5)] if there is an induced $K_4$ in $G_X$ and $U$ is an independent set in $G$, then for every vertex $v\in \{u_{12},u_{13},u_{14},u_{23},u_{24},u_{34}\}$, we have $|B_2(v)\cap X|\le 3$;
\item[(B6)] if there is an induced $K_4$ in $G_X$ and $U$ is an independent set in $G$, then the vertices of at most one pair among $\{(u_{12},u_{34}),$ $(u_{13},u_{24}),(u_{14},u_{23})\}$ have a common neighbor.
\end{enumerate}
\textit{Proof.}

 (B1): Let $u$, $v$ and $w$ be three distinct vertices of $X$ and let $z$ be a common neighbor of these three vertices.
Suppose $d_{G_X}(u)\ge 3$. Notice that, for any odd cycle $C_u$, $z\notin C_u$, since it would imply that either $v$ or $w$ belongs to $V(C_u)$. Thus, there exist two distinct vertices $u_1$ and $u_2$ such that $\{u_1,u_2\}\subseteq (N(u)\setminus \{z\})\cap V(C_u)$. Moreover, both $u_1$ and $u_2$ have at least two neighbors in $V(C_u)$. Since $d_{G_X}(u)\ge 3$, either $u_1$ or $u_2$ has a neighbor in $X$. Suppose, without loss of generality, that $u_1$ has a neighbor $u'\in X$. In this case, we can also suppose that there exists one odd cycle $C_{u'}$ such that $u_1\notin V(C_{u'})$. If such odd cycle did not exist, then $(X\setminus\{u,u'\})\cup\{u_1\}$ would be also an odd-cut set, contradicting the minimality of $X$.
Let $u''$ be the neighbor of $u_1$ in $V(C_u)\setminus\{u\}$. Since $u'$ has two neighbors in $V(C_{u'})$, $u'$ has no neighbors in $X$ and since $u''$ has two neighbors in $V(C_u)$, we obtain that $|(B_2(u_1)\cap X)\setminus \{u\}|\le 2$.
Consequently, $(X\setminus\{u\})\cup \{u_1\}$ is also a odd cut set, and since $\{u_1v ,u_1 w\}\notin E(G_{(X\setminus\{u\})\cup \{u_1\}})$ and we added at most one edge in $E(G_{(X\setminus\{u\})\cup \{u_1\}})$, we have $|E(G_{(X\setminus\{u\})\cup \{u_1\}})|<|E(G_X)|$, contradicting the minimality of $|E(G_X)|$.

(B2): Notice that by Property (B1), no three vertices of $U$ have a common neighbor. Moreover, by Property (A1), a vertex of $X$ cannot have two neighbors in $X$. Consequently, at most two pairs of vertices of $U$ are adjacent.
Suppose, without loss of generality, that $(u_1,u_2)$ and $(u_3,u_4)$ are two pairs of adjacent vertices.
In this case, we remove $U$ from $X$ and replace it by $\{u_{13},u_{14},u_{24}\}$, contradicting the minimality of $X$.

(B3): Suppose, without loss of generality, that $u_{12}$ and $u_{13}$ exist and are adjacent. Let $u'$ and $u''$ be two existing vertices among $\{u_{14},u_{23},u_{24},u_{34}\}$ such that the three vertices $u'$, $u''$ and $u_{13}$ have no common neighbor. In this case, we remove $U$ from $X$ and replace it by $\{u',u'',u_{13}\}$, contradicting the minimality of $X$. There exist three configurations depending on which are the two adjacent vertices from $U$. However, for the three possibilities, it is trivial to check that if we remove $U$ from $X$ and replace it by $\{u',u'',u_{13}\}$ we obtain an odd-cut set.

(B4): Suppose, without loss of generality, that $u_1$ and $u_2$ are adjacent. We can remark that, by Property (B1), no vertex of $\{u_{13},u_{14},u_{23},u_{24},u_{34}\}$ has a neighbor in $X\setminus U$. A consequence is that $|B_2(u_{34})\cap X|\le 4$. Suppose that $|B_2(u_{13})\cap X|> 4$. Then, $X$ contains two neighbors of $u'_{13}$. Consequently, every odd cycle $C$ with $u_1\in V(C)$ contains also another vertex of $X$. Hence, $X\setminus\{u_1\}$ is an odd-cut set, contradicting the minimality of $X$. By symmetry, the same goes for $u_{14}$, $u_{23}$ and $u_{24}$.

(B5): Suppose that the vertex $u_{12}$ satisfies $|B_2(u_{12})\cap X|> 3$. We can remark that if $u_{13}$ satisfied also $|B_2(u_{13})\cap X|> 3$, then, as for Property (B4), $X\setminus\{u_1\}$ would be an odd-cut set, contradicting the minimality of $X$. Thus, we can suppose that  $|B_2(u_{13})\cap X|\le 3$. Consequently, we can remove $\{u_1\}$ from $X$ and replace it by $\{u_{13}\}$, contradicting the minimality of $|E(G_X)|$. By symmetry, the same property holds for $u_{13}$, $u_{14}$, $u_{23}$, $u_{24}$ and $u_{34}$.

(B6): Suppose that the vertices of two pairs among $\{(u_{12},u_{34}),(u_{13},u_{24}),(u_{14},u_{23})\}$ have a common neighbor and that these two pairs are $(u_{12},u_{34})$ and $(u_{13},u_{24})$, i.e. $N(u_{12})\cap N(u_{34})\neq \emptyset$ and $N(u_{13})\cap N(u_{24})\neq \emptyset$ (other pairs can be treated similarly). Notice that $(X\setminus U)$ $\cup \{u'_{12},u_{14},u_{23}\}$ is also an odd-cut set, contradicting the minimality of $X$. Note that in this case, there is an even cycle containing each vertex of $U$ that does not contain vertices of $(X\setminus U)\cup\{u'_{12},u_{14},u_{23}\}$.
\begin{flushright} $\square$ \end{flushright}

 \begin{figure}
 \begin{center}
 \begin{tikzpicture}[scale=0.8]
\node (v1) at (-2-1,0.5) [draw,circle,color=blue,fill=blue,scale=0.7] {};
\node (w1) at (2-1,0.5) [draw,circle,color=blue,fill=blue,scale=0.7] {};
\node (z1) at (-0.5-1,2)[draw,circle,color=blue,fill=blue,scale=0.7] {};
\node (z2) at (0.5-1,2) [draw,circle,color=blue,fill=blue,scale=0.7] {};

 \node (u1) at (5,-2+3) [draw,circle,color=blue,fill=blue,scale=0.7] {};
 \node (u2) at (5,0+3) [draw,circle,color=blue,fill=blue,scale=0.7] {};
 \node (u3) at (7,-2+3) [draw,circle,color=blue,fill=blue,scale=0.7] {};
 \node (u4) at (7,0+3) [draw,circle,color=blue,fill=blue,scale=0.7] {};

 \node (i1) at (5,1-3) [draw,circle,color=blue,fill=blue,scale=0.7] {};
 \node (i2) at (5,3-3) [draw,circle,color=blue,fill=blue,scale=0.7] {};
 \node (i3) at (7,1-3) [draw,circle,color=blue,fill=blue,scale=0.7] {};
 \node (i4) at (7,3-3) [draw,circle,color=blue,fill=blue,scale=0.7] {};

 \node (x12) at (9+2,-1+3) [draw,circle,color=red,fill=red,scale=0.7]{};
 \node (x34) at (11+2,-1+3) [draw,circle,color=red,fill=red,scale=0.7]{};
 \node (x13) at (10+2,-2+3) [draw,circle,color=red,fill=red,scale=0.7]{};
 \node (x24) at (10+2,0+3) [draw,circle,color=red,fill=red,scale=0.7]{};
 \node at (8+1,-1){$\rightarrow$};

 \node (y13) at (10+2,1-3) [draw,circle,color=red,fill=red,scale=0.7]{};
 \node (y2) at (9+2,3-3) [draw,circle,color=red,fill=red,scale=0.7] {};
 \node (y3) at (11+2,1-3) [draw,circle,color=red,fill=red,scale=0.7] {};
 \node (y4) at (11+2,3-3) [draw,circle,color=red,fill=red,scale=0.7] {}; 
\node at (8+1,2){$\rightarrow$};

 \tikzstyle{every node}=[draw,circle,fill=white,minimum size=3pt,inner sep=1pt]
 \draw (-1-1,0) node (v)[]  {};
 \draw (0-1,0) node (u) [] {$u$};
 \draw (1-1,0) node (w) [] {};
 \draw (0-1,1) node (z) [] {};
 \draw (-2-1,-0.5) node (v2)[]  {};
 \draw (2-1,-0.5) node (w2)[]  {};
 \draw[ultra thick,color=blue] (u) -- (v);
 \draw[ultra thick,color=blue] (u) -- (w);
 \draw (u) -- (z);
 \draw (v) -- (v1);
 \draw[ultra thick,color=blue] (v) -- (v2);
 \draw (w) -- (w1);
 \draw[ultra thick,color=blue] (w) -- (w2);
 \draw (z) -- (z1);
 \draw (z) -- (z2);

 \draw (5,-1+3) node (u12) [] {};
 \draw (7,-1+3) node (u34) [] {};
 \draw (6,-2+3) node (u13) [] {};
 \draw (6,0+3) node (u24) [] {};
 \draw (5.8,-1+3) node (u14) [] {};
 \draw (6.2,-1+3) node (u23) [] {};
\draw (u1) -- (u12);
\draw (u2) -- (u12);
\draw (u3) -- (u34);
\draw (u4) -- (u34);
\draw (u1) -- (u13);
\draw (u1) -- (u14);
\draw (u2) -- (u23);
\draw (u2) -- (u24);
\draw (u3) -- (u13);
\draw (u3) -- (u23);
\draw (u4) -- (u14);
\draw (u4) -- (u24);

 \draw (7,2-3) node (i34) [] {};
 \draw (6,1-3) node (i13) [] {};
 \draw (6,3-3) node (i24) [] {};
 \draw (5.8,2-3) node (i14) [] {};
 \draw (6.2,2-3) node (i23) [] {};
\draw (i1) -- (i2);
\draw (i3) -- (i34);
\draw (i4) -- (i34);
\draw (i1) -- (i13);
\draw (i1) -- (i14);
\draw (i2) -- (i23);
\draw (i2) -- (i24);
\draw (i3) -- (i13);
\draw (i3) -- (i23);
\draw (i4) -- (i14);
\draw (i4) -- (i24);

 \draw (9+2,-2+3) node (x1) [] {};
 \draw (9+2,0+3) node (x2) [] {};
 \draw (11+2,-2+3) node (x3) [] {};
 \draw (11+2,0+3) node (x4) [] {};
 \draw (9.8+2,-1+3) node (x14) [] {};
 \draw (10.2+2,-1+3) node (x23) [] {};
\draw (x1) -- (x12);
\draw (x2) -- (x12);
\draw (x3) -- (x34);
\draw (x4) -- (x34);
\draw (x1) -- (x13);
\draw (x1) -- (x14);
\draw (x2) -- (x23);
\draw (x2) -- (x24);
\draw (x3) -- (x13);
\draw (x3) -- (x23);
\draw (x4) -- (x14);
\draw (x4) -- (x24);
 \draw (9+2,1-3) node (y1) [] {};
 \draw (9.8+2,2-3) node (y14) [] {};
 \draw (10.2+2,2-3) node (y23) [] {};
 \draw (11+2,2-3) node (y34) {};
 \draw (10+2,3-3) node  (y24){};
\draw (y1) -- (y2);
\draw (y3) -- (y34);
\draw (y4) -- (y34);
\draw (y1) -- (y13);
\draw (y1) -- (y14);
\draw (y2) -- (y23);
\draw (y2) -- (y24);
\draw (y3) -- (y13);
\draw (y3) -- (y23);
\draw (y4) -- (y14);
\draw (y4) -- (y24);
 \end{tikzpicture}
 \end{center}
\caption{\label{f1p11222a} A vertex $u\in X$ satisfying $d_{G_X}(u)=4$ (on the left) and four vertices of $X$ forming an induced $K_4$ in $G_X$ (on the middle) and their corresponding vertices in $Y$ (on the right) (bold vertices: vertices from $X$ (on the left and middle) or $Y$ (on the right), thick edges: edges belonging to $C_u$).}
\end{figure}
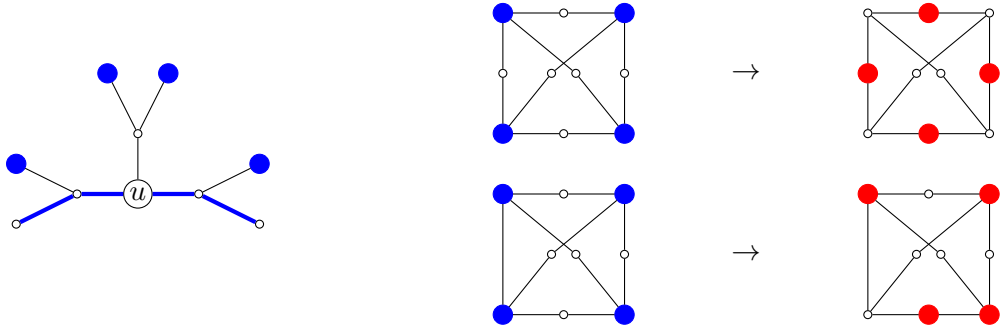

Now suppose that $X$ is an odd-cut set in $G$ of minimum cardinality minimizing $|E(G_X)|$.
We consider two types of induced $K_4$ in $G_X$. The first type is a $K_4$ such that $U$ is an independent set in $G$ and that both $(u_{12},u_{34})$ and $(u_{14},u_{23})$ are pairs of vertices with no common neighbor.
The second type is a $K_4$ such that $u_1$ and $u_2$ are adjacent in $G$.
By Properties (B2), (B3) and (B6), every induced $K_4$ in $G_X$ is isomorphic to an induced $K_4$ of the first or second type. 
We construct $Y$ by replacing $U$ by $\{u_{12},u_{23},u_{34},u_{41}\}$ in $X$, each time there is an induced $K_4$ of the first type in $G_X$ and by replacing $U$ by $\{u_{2},u_{3},u_{4},u_{13}\}$ in $X$, each time there is an induced $K_4$ of the second type in $G_X$. The middle part of Figure \ref{f1p11222a} illustrates the two types of induced $K_4$ in $G_X$ and the right part illustrates the corresponding vertices in $Y$. Notice that, by Properties (B4) and (B5), the graph $G_Y$ has no more induced $K_4$.

Now, it remains to prove that $Y$ is an odd-cut set and that $G_Y$ is subcubic.
Suppose there is an induced $K_4$ in $G_X$.
If this $K_4$ is of the first type, then, since the set $\{u_{12},u_{23},u_{34},u_{41}\}$ contains two neighbors of each vertex from $U$, we obtain that if $u_i$ belongs to an odd cycle, then at least one vertex from $\{u_{12},u_{23},u_{34},u_{41}\}$ belongs to this odd cycle. If this $K_4$ is of the second type, notice that if $u_1$ belongs to an odd cycle, then at least one vertex from $\{u_2,u_{13}\}$ belongs to this odd cycle.
Consequently, $Y$ is an odd-cut set and since $X$ has the same cardinality than $Y$, $Y$ has minimum cardinality and satisfies Property (A4). Moreover, by Properties (B4) and (B5), we obtain that no vertex of $Y\setminus X$ has a common neighbor with two other vertices of $Y$ and consequently that $Y$ satisfies Property (B1). Hence, by Property (A4), we obtain that $G_{Y}$ is subcubic and has no connected component isomorphic to $K_4$, which concludes the proof.

\end{proof}

Remark that in the previous proof, except in the proof of Property (B6), we did not use the fact that the cycles have odd length. Thus, using a similar proof, we can possibly obtain that for any subcubic graph $G$, there exists a minimum cut set that can be partitioned in three disjoint $2$-packings of $G$.

The Petersen graph is an example of cubic graph which is not $(1,1,k,k')$-colorable, for any $k,k'\ge2$, showing that the result of Theorem~\ref{p13} is tight in a certain sense.

\begin{prop}
For any $k,k'\ge2$, the Petersen graph is not $(1,1,k,k')$-colorable.
\end{prop}
\begin{proof}
Since the diameter of the Petersen graph is $2$, for any $k\ge 2$, a color $k$ can be used to color only one vertex.
Hence, it suffices to prove that we can color at most seven vertices with the two colors $1$.
A maximum independent set in the Petersen graph contains at most four vertices and an independent set of four vertices is made of two vertices of the outer cycle and two vertices of the inner cycle of the Petersen graph.
We can remark that for every possible independent set of four vertices $A$, there are at most three non-adjacent vertices not belonging to $A$.
\end{proof}

The experiments suggest that the Petersen graph could be the only non $(1,1,2,3)$-colorable subcubic graph, see Table~\ref{t1}. Furthermore, the next result shows that the three colors 2 cannot be replaced by three colors 3 in Theorem~\ref{p13}.
\begin{prop}
There exist cubic graphs that are not $(1,1,3,3,3)$-colorable.
\end{prop}
\begin{proof}
Consider the cubic graph depicted in Figure~\ref{f1p16}.
Since it has diameter $3$, no more than one vertex could be colored by a color $3$.
Moreover, it contains four disjoint triangles and each triangle should contain one vertex not colored by $1$.
Thus, it is impossible to color it with the sequence $(1,1,3,3,3)$.
\end{proof}

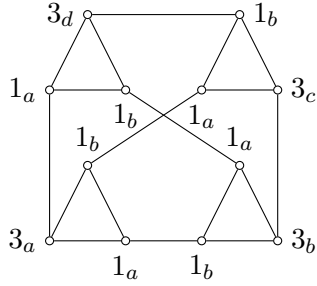
\begin{figure}
\begin{center}
\begin{tikzpicture}
\tikzstyle{every node}=[draw,circle,fill=white,minimum size=3pt,inner sep=1pt]
\draw (0,0) node (x1)[label=left:$3_a$]  {};
\draw (1,0) node (x2) [label=south:$1_a$] {};
\draw (2,0) node (y1) [label=south:$1_b$] {};
\draw (3,0) node (y2)[label=right:$3_b$]  {};
\draw (0.5,1) node (x3)[label=north:$1_b$]   {};
\draw (2.5,1) node (y3)[label=north:$1_a$]  {};
\draw (0,2) node (v1)[label=left:$1_a$]  {};
\draw (1,2) node (v2) [label=south:$1_b$] {};
\draw (2,2) node (w1) [label=south:$1_a$] {};
\draw (3,2) node (w2)[label=right:$3_c$]  {};
\draw (0.5,3) node (v3)[label=left:$3_d$]   {};
\draw (2.5,3) node (w3)[label=right:$1_b$]  {};
\draw (x1) -- (x2);
\draw (x2) -- (x3);
\draw (x3) -- (x1);
\draw (y1) -- (y2);
\draw (y2) -- (y3);
\draw (y3) -- (y1);
\draw (v1) -- (v2);
\draw (v2) -- (v3);
\draw (v3) -- (v1);
\draw (w1) -- (w2);
\draw (w2) -- (w3);
\draw (w3) -- (w1);
\draw (x2) -- (y1);
\draw (x1) -- (v1);
\draw (y3) -- (v2);
\draw (y2) -- (w2);
\draw (x3) -- (w1);
\draw (w3) -- (v3);

\end{tikzpicture}
\end{center}
\caption{\label{f1p16} A cubic non $(1,1,3,3,3)$-colorable graph of order 12.}
\end{figure}

\begin{table}[!ht]
\centering\begin{tabular}{c|cccc}
$n\backslash S$ &$(1,1)$&$(1,1,2)$&$(1,1,2,3)$&$(1,1,2,3,3)$\\\hline
4&  0 & 0 & 1 & 0\\
6&  1 & 0 & 1 & 0\\
8&  1&  2&  2&  0\\
10&  2&  9&  7&  1\\
12&  5&  42&  38&  0\\
14&  13&  314&  182&  0\\
16&  38&  2808&  1214&  0\\
18&  149&  32766&  8386&  0\\
20&  703&  423338&  86448&  0\\
22&  4132&  6212201&  1103114&  0\\
\end{tabular}
\caption{\label{t1}Number of $S$-chromatic cubic graphs of order $n$ up to 22.}
\end{table}

We now show that 3-irregular subcubic graphs are $(1,1,2)$-colorable. Notice that the subdivided graph $S(G)$ of any graph $G$ is $(1,1)$-colorable as it is bipartite.
\begin{Theorem}\label{p14}
Every 3-irregular subcubic graph is $(1,1,2)$-colorable. 
\end{Theorem}
\begin{proof}
Let $G$ be a 3-irregular graph and let $e=xy$ be any edge of $G$ such that $x$ and $y$ both have degree at most $2$.
If no such edge exist then the graph is bipartite and consequently $(1,1)$-colorable.

Define a level ordering $L_i$, $0\le i \le r=\epsilon(e)$, of $(G,e)$.

We first construct a coloring $c$ of the vertices of $G$ from level $r$ to $1$ and with colors from the set $\{1_a,1_b,2\}$, that satisfies the following property:
\begin{enumerate}
\item[i)] No vertex of degree at most 2 is colored 2.
\end{enumerate}

The set $L_{r}$ induces a disjoint union of paths of order at most 3 in $G$. Since paths are $(1,1)$-colorable, $L_{r}$ is $(1,1)$-colorable. Thus, Property i) is satisfied.

Assume that we have already colored all vertices of $G$ of levels from $r$ to $i+1$ and that we are going to color vertex $u\in L_i$, $1\le i\le r-1$. 
If $u$ has degree at most $2$, then $\{1_a,1_b\}\not \subseteq C_1(u)$. Hence, $u$ can be colored by $1_a$ or $1_b$ and Property i) is satisfied.
If $u$ has degree $3$ and if $C_1(u)\neq \{1_a,1_b\}$, then $u$ can be colored by $1_a$ or $1_b$.
Else if $C_1(u)=\{1_a,1_b\}$, let $u_1$ and $u_2$ be the colored neighbors of $u$, with $c(u_1)=1_a$ and $c(u_2)=1_b$.
The vertex $u_1$ has a neighbor colored by $1_b$ and the vertex $u_2$ has a neighbor colored by $1_a$, if not $u_1$ and $u_2$ could be recolored and $u$ could be colored by $1_a$ or $1_b$.
Thus, we have $C_1(u)\cup C_2(u)=\{1_a,1_b\}$ because $G$ is 3-irregular and we can color $u$ by the color $2$.

Finally, it remains to color vertices of $L_0$, i.e. $x$ and $y$. If $1_a\not\in C_1(x)$ and $1_b\not\in C_1(y)$ (or, symmetrically, $1_b\not\in C_1(x)$ and $1_a\not\in C_1(y)$). Then set $c(x)=1_a$ and $c(y)=1_b$ (or, symmetrically, $c(x)=1_b$ and $c(y)=1_a$). 
Let $x_1$ be the possible neighbor of $x$ different from $y$ and let $y_1$ be the possible neighbor of $y$ different from $x$.
Without loss of generality, suppose that $c(x_1)=1_a$ and $c(y_1)=1_a$.
Suppose that $x_1$ has degree at most 2. If $2\in C_1(x_1)$, then $x_1$ can be recolored by $1_b$ and we can set $c(x)=1_a$ and $c(y)=1_b$.
Else, $C_1(x_1)=\{1_b\}$ and we can set $c(x)=2$ and $c(y)=1_b$. If $x_1$ has degree 3, then every colored neighbor of $x$ has at most degree 2 and is colored by $1_b$ by Property i).
Thus, since $2\notin C_1(x_1)$, we can set $c(x)=2$ and $c(y)=1_b$.
Therefore, we obtain a $(1,1,2)$-coloring of $G$.
\end{proof}

\section{$(1,2,3,\ldots)$-coloring}
The question of whether cubic graphs have finite packing chromatic number or not was raised by Goddard et al.~\cite{GoBro}.
We give some partial results related to this question.

For the subdivision of a cubic graph,  Proposition~\ref{p01} implies that if every subcubic graph $G$ different from the Petersen graph is $(1,1,2,2)$-colorable, then $S(G)$ is $(1,3,3,5,5)$-colorable and consequently $\chi_\rho(S(G))\le 5$. On the other side, it can be easily verified that $\chi_\rho(S(K_4))=5$.

\begin{figure}
\begin{center}
\begin{tikzpicture}
\tikzstyle{every node}=[draw,circle,fill=white,minimum size=3pt,inner sep=1pt]
\draw (0,0) node (x1)[] {};
\draw (1.2,0) node (x10)[] {};
\draw (0.3,0.6) node (x2)[] {};
\draw (0.8,1) node (x3)[] {};
\draw (1.6,1) node (x4)[] {};
\draw (2.4,0.3) node (x5)[] {};

\draw (0.3,-0.6) node (x6)[] {};
\draw (0.8,-1) node (x7)[] {};
\draw (1.6,-1) node (x8)[] {};
\draw (2.4,-0.3) node (x9)[] {};

\draw (3,0.3) node (y1)[] {};
\draw (3.3+0.2,1.2) node (y2)[] {};
\draw (3.8+0.4,1.5) node (y3)[] {};
\draw (4.2+0.6,1.7) node (y4)[] {};
\draw (4.6+0.8,1.5) node (y5)[] {};
\draw (5.1+1,1.2) node (y6)[] {};
\draw (5.4+1.2,0.3) node (y7)[] {};

\draw (3,-0.3) node (y8)[] {};
\draw (3.3+0.2,-1.2) node (y9)[] {};
\draw (3.8+0.4,-1.5) node (y10)[] {};
\draw (4.2+0.6,-1.7) node (y11)[] {};
\draw (4.6+0.8,-1.5) node (y12)[] {};
\draw (5.1+1,-1.2) node (y13)[] {};
\draw (5.4+1.2,-0.3) node (y14)[] {};

\draw (x1) -- (x10);
\draw (x1) -- (x2);
\draw (x2) -- (x3);
\draw (x3) -- (x4);
\draw (x4) -- (x5);
\draw (x1) -- (x6);
\draw (x6) -- (x7);
\draw (x7) -- (x8);
\draw (x8) -- (x9);
\draw (x2) -- (x9);
\draw (x6) -- (x5);
\draw (x3) -- (x7);
\draw (x10) -- (x8);
\draw (x10) -- (x4);

\draw (x5) -- (y1);
\draw (y1) -- (y2);
\draw (y2) -- (y3);
\draw (y3) -- (y4);
\draw (y4) -- (y5);
\draw (y5) -- (y6);
\draw (y6) -- (y7);
\draw (y7) -- (y14);
\draw (y8) -- (y9);
\draw (y9) -- (y10);
\draw (y10) -- (y11);
\draw (y11) -- (y12);
\draw (y12) -- (y13);
\draw (y13) -- (y14);
\draw (x9) -- (y8);

\draw (y1) -- (y5);
\draw (y2) -- (y7);
\draw (y3) -- (y13);
\draw (y4) -- (y11);
\draw (y6) -- (y10);
\draw (y9) -- (y14);
\draw (y8) -- (y12);

\end{tikzpicture}
\end{center}
\caption{\label{f1p17} A cubic graph of order $24$ with packing chromatic number 11.}
\end{figure}
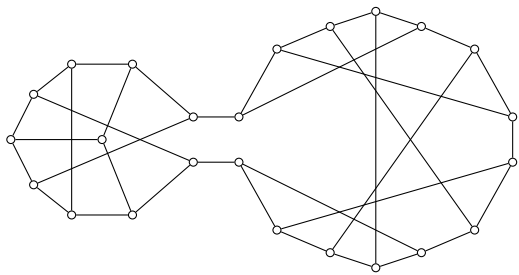

For arbitrary cubic graphs, we can (only) state the following:
\begin{prop}There exists a cubic graph with packing chromatic number $13$.
\end{prop}

\begin{proof}
The cubic graph of order $38$ and diameter 4 (which is a largest cubic graph with diameter 4) described independently in~\cite{AFY86,vonC} needs 13 colors to be packing colored (checked by computer). By running a brute force search algorithm, we found that at most $28$ vertices can be colored with colors $\{1,2,3\}$. But, since this graph has diameter 4, then every color greater than 3 can be given to only one vertex, implying the use of all colors from $\{4,\ldots,13\}$ to complete the coloring.
\end{proof}

The distribution of packing chromatic numbers for cubic graphs of order up to $20$ is presented in Table~\ref{tb3}.
With the help of a computer, we also found a cubic graph of order $24$ and packing chromatic number $11$.
This graph is illustrated in Figure~\ref{f1p17}.

\begin{table}[!ht]
\centering\begin{tabular}{c|cccccccc}
$n$ $\backslash$ $\chi_{\rho}$ & 4 & 5 & 6 & 7 & 8 & 9 & 10 & 11\\\hline
4  & 1 & 0  & 0   & 0 & 0 & 0 & 0 & 0\\
6  & 1 & 1  & 0   & 0 & 0 & 0 & 0 & 0\\
8  & 0 & 3  & 2   & 0 & 0 & 0 & 0 & 0\\
10 & 0 & 3  & 15  & 1   & 0 & 0 & 0 & 0\\
12 & 0 & 7  & 42  & 36  & 0 & 0 & 0 & 0\\
14 & 0 & 13  & 252 & 222 & 22 & 0 & 0 & 0\\
16 & 0 & 34 & 907 & 2685 &  433 & 1 & 0 & 0\\
18 & 0 & 116 & 5277 & 21544 & 14050 &  314 & 0 & 0\\
20 & 0 & 151	& 22098	& 206334& 226622 & \multicolumn{2}{c}{55284$^*$}&0\\
\end{tabular}
\caption{\label{tb3}Number of cubic graphs of order $n$ with packing chromatic number $\chi_{\rho}$ up to 20.$^*$There are 55284 cubic graphs of order 20 and with packing chromatic number between 9 and 10 (our program takes too long time to compute their packing chromatic numbers).}
\end{table}

%

\section{Concluding remarks}
We conclude this paper by listing a few open problems:
\begin{itemize}
\item Is it true that any subcubic graph except the Petersen graph is $(1,1,2,3)$-colorable?
\item Is it true that any subcubic graph except the Petersen graph is $(1,2,2,2,2,2)$-colorable?
\item Does there exist a 3-irregular subcubic graph that is not $(1,2,2,3)$-colorable?
\item Is it true that any 3-irregular subcubic graph is $(1,1,3)$-colorable?
\item Is it true that the subdivision of any subcubic graph is $(1,2,3,4,5)$-colorable?
\item Does there exist a cubic graph with packing chromatic number larger than 13?
\end{itemize}
\section*{Acknowledgments} 
First author was partially supported by the Burgundy Council under grant \#CRB2011-9201AAO048S05587.


\end{document}